\documentclass[pra,10pt,twocolumn,superscriptaddress,floatfix,nofootinbib,showpacs]{revtex4-1}

\usepackage[utf8]{inputenc}  
\usepackage[T1]{fontenc}     
\usepackage[british]{babel}  
\usepackage[sc,osf]{mathpazo}
\usepackage[scaled=0.86]{berasans}  
\usepackage[colorlinks=true, citecolor=blue, urlcolor=blue]{hyperref}  
\usepackage{scrextend}
\usepackage{graphicx} 
\usepackage[babel]{microtype}  
\usepackage{amsmath,amssymb,amsthm,bm,amsfonts,mathrsfs,bbm} 

\usepackage{xspace}  
\usepackage{xcolor,colortbl}
\usepackage{array}
\usepackage{bigstrut}
\usepackage{braket}
\usepackage{color}
\usepackage{natbib}

\newcommand{\be}{\begin{equation}}
\newcommand{\ee}{\end{equation}}
\newcommand{\ba}{\begin{eqnarray}}
\newcommand{\ea}{\end{eqnarray}}

\newtheorem{theorem}{Theorem}

\newtheorem{Lemma}{Lemma}

\begin{document}

\title{Local marginals ameliorate device independent witnessing of genuine entanglement}

\author{Some Sankar Bhattacharya}
\affiliation{Physics and Applied Mathematics Unit, Indian Statistical Institute, 203 B. T. Road, Kolkata 700108, India.}

\author{Biswajit Paul}
\affiliation{Department of Mathematics, South Malda College, Malda, West Bengal, India.}

\author{Arup Roy}
\affiliation{Physics and Applied Mathematics Unit, Indian Statistical Institute, 203 B. T. Road, Kolkata 700108, India.}

\author{Amit Mukherjee}
\affiliation{Physics and Applied Mathematics Unit, Indian Statistical Institute, 203 B. T. Road, Kolkata 700108, India.} 

\author{C. Jebaratnam}
\affiliation{S. N. Bose National Centre for Basic Sciences, Salt Lake, Kolkata 700 098, India}

\author{Manik Banik}
\affiliation{Optics \& Quantum Information Group, The Institute of Mathematical Sciences, HBNI, C.I.T Campus, Tharamani, Chennai 600 113, India.}


\begin{abstract}
We consider the problem of determining the presence of genuine multipartite entanglement through the violation of Mermin's Bell-type inequality (MI). Though the violation of MI cannot certify the presence of genuine nonlocality, but can certify genuine tripartite entanglement whenever the violation is strictly greater than $2\sqrt{2}$. 
Here we show that MI suffices as genuine entanglement witness even when its value is $2\sqrt{2}$ if at least two of the local marginal distributions are not completely random provided the local Hilbert space dimension of at least one of the sub-systems is two. Thus local marginals suffice as semi-device independent genuine entanglement witness. 
This is intriguing in a sense, as the local properties of a composite system can help to identify its global property. Furthermore, analyzing another quantity constructed from Mermin polynomials we show that genuine entanglement certification task for the correlations with MI violation equal to $2\sqrt{2}$ can actually be made fully device independent. 
       
\end{abstract}


\maketitle


Entanglement  is considered to be one of the most bizarre nonclassical manifestations of multipartite quantum systems. Linearity of quantum mechanics (QM) allows to build superposed states that cannot be written as product of states of each subsystems and hence resulting in entangled states. Entanglement lies at the core of some of the most puzzling features of QM: the Einstein-Podolski-Rosen (EPR) argument \cite{EPR}, the Schrodinger's \emph{steering} concept \cite{schr}, and most importantly the \emph{nonlocal} behaviour of QM \cite{Bell1964}. In the past three decades quantum entanglement has also been established as a useful resource for several information theoretic tasks: quantum cryptography \cite{crypto}, quantum teleportation \cite{Bennett93}, quantum superdense coding \cite{Bennett92} are few noteworthy among many others.   

While the state of a quantum system composed of only two subsystems can be either separable or entangled, for a quantum system with more than two subsystems the separability properties have a complicated structure. In the multipartite scenario the most dramatic attribute appearing in the picture is the concept of truly or genuinely multipartite entanglement (GME) which cannot be prepared by mixing states that are separable with respect to some partition \cite{GHZ,Dur,Horodecki'09,Guh}. Whereas Greenberger, Horne, and Zeilinger established implication of GME in quantum foundations by reveling \emph{perfect} incompatibility of QM with EPR idea of \emph{local deterministic} world view \cite{GHZ}, significant advantages of it compared to bipartite entanglement has also been established in different practical tasks, namely, extreme spin squeezing \cite{Sor}, high sensitivity in some general metrology tasks \cite{Hyl}, quantum computing using cluster states \cite{Rau}, measurement-based quantum computation \cite{Bri}, various quantum communication protocols \cite{Mur}, secret sharing among multiple parties and multiparty quantum network \cite{Buz}. Despite its importance, characterization and detection of entanglement is quite difficult. Several methods such as tomography of the full state\cite{Alt}, constructing linear and/or nonlinear entanglement witnesses \cite{Guh}, or observing the violation of Bell-like inequalities\cite{Mor}, have been proposed for verification or certification of presence of GME. However we are still far from understanding multipartite entanglement completely.  

Certifying the presence of GME by some nonlocality based arguments, i.e., observing violation of some Bell type inequalities has an advantage over full state tomography or constructing usual witness operators. Whereas in the later two methods one must require perfect experimental devices, contrarily, nonlocality based arguments are dependent on observed measurement statistics only, without relying on the types of measurements performed, the precision involved in their implementation, or on assumptions about the relevant Hilbert space dimension. Such a certification of GME witness is known as device-independent (DI) certification. Several results have been reported in the recent past where the question of witnessing  DI-GME has been addressed \cite{See,Nag,Uff,SI,Sev,Ban,Pal}. The basic idea of DI certification is to check whether outcome statistics of different measurements performed locally by each party on the given unknown multipartite state violate some genuine multipartite bell-type inequality or not.
For the tripartite case one example of such inequality is the well known Svetlichny’s inequality \cite{SI}. If a tripartite input-output correlation violates Svetlichny’s inequality then the correlation is genuinely nonlocal and the quantum state producing such correlation must be genuinely entangled. However, the relation between DI genuine entanglement witness (DI-GEW) and multipartite nonlocality is more subtle. For example, violation of Mermin's Bell-type inequality (MI) \cite{Mer}, i.e., $\langle \mathbb{M}\rangle\le 2$, cannot be used to certify the presence of genuine tripartite nonlocality, but, it can reveal tripartite genuine entanglement whenever $\langle \mathbb{M}\rangle>2\sqrt{2}$ \cite{See}. The MI therefore being a two-way nonlocal witness can suffice as DI-GME. In \cite{Ban}, Bancal et al. have also introduced an $n$-partite inequality which cannot reveal genuine multipartite nonlocality but can suffice as DI witness for GME.

In this work we explore the possibility of witnessing genuine entanglement from MI even when its value is $\langle M\rangle=2\sqrt{2}$. We find that tripartite genuine entanglement can be certified for $\langle \mathbb{M}\rangle=2\sqrt{2}$ if at least two of the local marginals are not completely random, provided at least one of the subsystems dimensions of the tripartite system is two. In other word, local marginals can help in semi-device independent witnessing of genuine entanglement for Mermin value $\langle \mathbb{M}\rangle=2\sqrt{2}$. We then show that analyzing an other quantity, obtained from the correlation only and introduced recently in \cite{Jebaratnam} by one of our authors, it is even possible to witness the tripartite genuine entanglement in a fully device independent manner. It is intriguing in a sense that the local statistics can help to certify a global property of the system (here the genuineness of entanglement) in device independent manner. Before going to our main result we first briefly describe tripartite quantum and no-signalling correlations. 

\emph{Tripartite quantum system}.--- Consider an arbitrary tripartite system $\rho_{ABC}\in\mathcal{D}(\mathbb{C}^{d_A}\otimes\mathbb{C}^{d_B}\otimes\mathbb{C}^{d_C})$ shared among three parties (say) Alice, Bob and Charlie, where $\mathcal{D}(\mathbb{X})$ denotes the set of density operator acting on Hilbert space $\mathbb{X}$ and $d_m$ be the local Hilbert space dimension of the $m^{th}$ party, for $m=A,B,C$. To classify the type of entanglement present in $\rho_{ABC}$ we need to consider all possible pure state decompositions of the given state, i.e., $\rho_{ABC}=\sum_jp_j|\psi^j_{ABC}\rangle\langle\psi^j_{ABC}|$, with $\{p_j\}$ being a valid probability distribution and $|\psi^j_{ABC}\rangle\in\mathbb{C}^{d_A}\otimes\mathbb{C}^{d_B}
\otimes\mathbb{C}^{d_C}$, for all $j$. A state $\rho_{ABC}$ is called fully separable if there exists a decomposition where $\forall~j$ we have $|\psi^j_{ABC}\rangle=|\psi^j_{A}\rangle\otimes|\psi^j_{B}\rangle\otimes|\psi^j_{C}\rangle$. If all $|\psi^j_{ABC}\rangle$ can be
written as either $|\psi^j_{AB}\rangle\otimes|\psi^j_{C}\rangle$ or |$\psi^j_{AC}\rangle\otimes|\psi^j_{B}\rangle$ or $|\psi^j_{BC}\rangle\otimes|\psi^j_{A}\rangle$ and
at least one of them is not a product state, then the state is
entangled, but there is no genuine three-particle entanglement. Such state are also called bi-separable and can be expressed as,
\begin{eqnarray}\label{bi-sep}
\rho_{ABC}^{bi.sep}&=&\underset{i}\sum p_i \rho^i_A \otimes \rho^i_{BC} + \underset{j}\sum q_j \rho^j_B \otimes \rho^j_{AC} + \underset{k}\sum r_k \rho^k_C \otimes \rho^k_{AB},\nonumber\\
&&\mbox{with}~ \sum_{i}p_i+\sum_{j}q_j+\sum_{k}r_k=1.
\end{eqnarray}   
Tripartite quantum state is genuinely entangled if it is neither fully separable nor bi-separable. GHZ class and W class are two canonical examples of genuine entanglement. 

\emph{Tripartite no-signaling correlation}.--- No-signaling (NS) scenario captures more general kind of correlations than obtained in quantum world. Let $x\in\mathcal{X}$, $y\in\mathcal{Y}$, and $z\in\mathcal{Z}$ denote inputs of Alice, Bob, and Charlie and $a\in\mathcal{A}$, $b\in\mathcal{B}$ and $c\in\mathcal{C}$ be their respective  outputs. Here, for our purpose all the inputs and outputs are considered to be two valued (however, the following discussion easily generalizes to higher valued input-output cases). A tripartite input-output probability distribution $P\equiv\{p(abc|xyz)~|~p(abc|xyz)\ge 0~\forall~x,y,z,a,b,c;~\&~\sum_{a,b,c}p(abc|xyz)=1,~\forall~x,y,z\}$ is called a NS probability distribution (denoted as $P^{NS}$) if it satisfies NS constraints, i.e., $\sum_a P(abc|xyz)=P(bc|yz)~\forall~ b,c,x,y,z$, with permutations of the parties, i.e, marginal distributions of any two parties are independent of the input chosen by the other party. A tripartite correlation is called fully local (denoted as $P^L$) if it can be decomposed as $P(abc|xyz)=\sum_\lambda p_\lambda P_\lambda(a|x)P_\lambda(b|y)P_\lambda(c|z)~\forall~a,b,c,x,y,z$, where $\{p_{\lambda}\}$ is again a probability distribution. On the other hand, a tripartite correlation is called two-way local (denoted as $P^{L_2}$) if it satisfies a decomposition of the form $P(abc|xyz)=p_1\sum_\lambda p_\lambda P_\lambda^{AB|C}+p_2\sum_\lambda q_\lambda P_\lambda^{AC|B}+p_3\sum_\lambda r_\lambda P_\lambda^{A|BC}$, where $P_\lambda^{AB|C}=P_\lambda(ab|xy)\,P_\lambda(c|z)$, with $P_{\lambda}(ab|xy)$ being some NS probability distribution between Alice and Bob, and, $P_\lambda^{AC|B}$ and $P_\lambda^{A|BC}$  are  similarly  defined. If a two-way local correlation does not allow a fully local decomposition then we say that it is two-way nonlocal. NS correlations having no two-way local decomposition are called genuinely nonlocal. The set of no-signaling, two-way local, and fully local correlations forms polytopes, respectively denoted as $\mathcal{NS}$, $\mathcal{L}_2$, and $\mathcal{L}$, that follow the strict set inclusion relations $\mathcal{L}\subset\mathcal{L}_2\subset\mathcal{NS}$. A tripartite correlation is called quantum if it can be expressed through the Born's rule, i.e.,
\be
P(abc|xyz)=\mathrm{Tr} \left[\rho_{ABC} \left(M^{x}_{a}\otimes M^{y}_{b}\otimes M^{z}_{c}\right)\right],
\ee
where $\rho_{ABC}\in\mathcal{D}(\mathbb{C}^{d_A}\otimes\mathbb{C}^{d_B}\otimes
\mathbb{C}^{d_C})$, and $\{M^x_a~|~M^x_a\ge 0~\forall~x,a;~\&~\sum_aM^x_a=\mathbf{1}_{d_A}\}$ is positive-operator-valued-measurement (POVM) on Alice's subsystem and similarly $\{M^y_b\}$ and $\{M^z_c\}$ are POVMs on Bob's and Charlie's subsystems. 

Correlations obtained from fully separable quantum states are always fully local and obtained from bi-separable quantum states are always two-way local \cite{Werner'01}. On the other hand, violation of genuine nonlocal inequality certifies presence of genuine nonlocal correlations and it also implies that the quantum state producing the correlation must be genuinely entangled one. Therefore, genuine nonlocal inequalities also suffice as DI-GEW.
Svetlichny's inequality is one such example \cite{SI}. Similarly, violation of two-way nonlocal inequality certifies presence of two-way nonlocal correlations and sufficiently guarantees present of entanglement in the quantum state. More interesting, violation of two-way nonlocal inequality can sometimes certifies presence of genuine entanglement in the state, which establishes the fact that relation between DI-GEW and witnesses of genuine nonlocality is more subtle. For example, consider the Mermin's Bell-type inequality,
\begin{equation}\label{Mermin}
|\langle x_0y_0z_1\rangle+\langle x_0y_1z_0\rangle+\langle x_1y_0z_0\rangle-\langle x_1y_1z_1\rangle|\le 2. 
\end{equation}
Violation of this inequality, i.e.,  $|\langle \mathbb{M}\rangle|:=|\langle x_0y_0z_1\rangle+\langle x_0y_1z_0\rangle+\langle x_1y_0z_0\rangle-\langle x_1y_1z_1\rangle|> 2$, certifies the presence of two-way nonlocal correlation, but it cannot reveals genuine nonlocality \cite{Ban2013}. Interestingly, as shown in \cite{See}, it can witness genuine  entanglement whenever $|\langle M\rangle|>2\sqrt{2}$. In the following, we show that even when the Mermin expression value is equal to $2\sqrt{2}$, a more meticulous analysis of the observed statistics, particularly marginal statistics, can help to certify the presence of genuine entanglement.         

\emph{Result}: We first discuss a lemma which is prerequisite to prove one of our results. 
\begin{Lemma}\label{lemma1}
(Wolf et al. \cite{Wolf2009}): Any two dichotomic quantum measurements (i.e. measurement with two outcomes) are not jointly measurable iff they violate the CHSH inequality.
\end{Lemma}
The authors in \cite{Wolf2009} have shown that for a pair of non jointly measurable observables (dichotomic) there exists a bipartite quantum state and a set of observables for an added site together with which the given observables violate a Bell inequality. They further argued that such a pair of incompatible quantum measurements cannot be measured jointly in any other NS theory. we use this result to prove the following theorem. 
\begin{theorem}\label{theorem1}
Bi-separable quantum states with atleast one of the local dimensions equal to two, and achieving the Mermin value $|\langle \mathbb{M} \rangle| = 2\sqrt{2}$ are of the form $\rho_{ABC} = |\psi_A\rangle\langle\psi_A| \otimes |\mathcal{B}_{BC}\rangle\langle\mathcal{B}_{BC}|$ (upto local isometry), where $|\mathcal{B}_{BC}\rangle$ is a bipartite maximally entangled states that maximizes $\langle CHSH\rangle_{BC}$, or analogous form with respect to party-permutation.
\end{theorem}
\begin{proof}
The Mermin operator of Eq.(\ref{Mermin}) can be rewritten as $M^{A:BC} = \frac{1}{2}[x_0 \otimes(CHSH_{BC} - CHSH^{'}_{BC}) + x_1 \otimes(CHSH_{BC} + CHSH^{'}_{BC})]$, where $CHSH_{BC} = y_1 z_1 + y_1 z_0 + y_0 z_1 - y_0z_0$ is the canonical Clauser-Horne-Shimony-Holt (CHSH) operator and $CHSH^{'}_{BC}$ is the same expression as $CHSH_{BC}$ with indices $0$ and $1$ interchanged. Similarly, Mermin operator can also be expressed in terms of$M^{B:AC}$ and $M^{C:AB}$ defined analogously. The expectation value of the Mermin operator (left hand side of Eq.(\ref{Mermin})) with respect to the bi-separable state (of Eq.(\ref{bi-sep})) thus become,
	\begin{widetext}
\begin{eqnarray}\label{bi1}
|\langle M \rangle| &=& |\frac{1}{2}(Tr(M^{A:BC}\underset{i}\sum p_i \rho^i_A \otimes \rho^i_{BC})
+Tr(M^{B:AC}\underset{j}\sum q_j \rho^j_B \otimes \rho^j_{AC})+Tr(M^{C:AB}\underset{i}\sum r_k \rho^k_C \otimes \rho^k_{AB}))| \nonumber\\
&=& |\underbrace{\frac{1}{2} (\underset{i}\sum p_i(\langle x_0 \rangle^i U^i_{BC} + \langle x_1 \rangle^i V^i_{BC}))}+\underbrace{\frac{1}{2} (\underset{j}\sum q_j(\langle y_0 \rangle^j U^j_{AC} + \langle y_1 \rangle^j V^j_{AC}))}+ \underbrace{\frac{1}{2}(\underset{k}\sum r_k(\langle z_0 \rangle^k U^k_{AB} + \langle z_1 \rangle^k V^k_{AB}))}|, 
\end{eqnarray}
		\end{widetext}
where $\langle x_0\rangle^i$ is the expectation value on the state $\rho_A^i$ and other single party expectations are defined analogously, and $U^i_{BC} = (\langle CHSH \rangle^i_{BC} - \langle CHSH \rangle^{'i}_{BC})$	and $V^i_{BC} = (\langle CHSH \rangle^i_BC + \langle CHSH \rangle^{'i}_{BC})$ with expectation defined on the state $\rho^i_{BC}$, and other $U$'s and $V$'s having analogous expressions. The maximum value that each of $|U^m_{l}|$ and $|V^m_{l}|$ $(l= BC,AC, AB; m = i, j, k)$ can achieve is $2\sqrt{2}$. This is because, the quantum measurement-setup giving maximum value (i.e. $2\sqrt{2}$) for CHSH (CHSH') gives zero value to CHSH' (CHSH) \cite{Su'2016}. So each of the three terms (individually) in the above expression can give the maximum value $2\sqrt{2}$ when each of $\langle x_n \rangle^i$, $\langle y_n \rangle^j$, $\langle z_n \rangle^k$ ($n=0,1$) is equal to $1$. Let us consider the first term on the right hand side of Eq.(\ref{bi1}). By Lemma-\ref{lemma1}, to violate CHSH local bound the measurements for both Bob and Charlie should be non-jointly measurable. On the other hand, for the optimal contribution from the second term in Eq.(\ref{bi1}) one must have $\langle y_n \rangle^j=1$ ($n=0,1$) i.e. the measurements on Bob's side should be jointly measurable, which can not be used to violate CHSH local bound. Similar arguments considering Alice and Charlie's measurement setting leads us to the following possibilities- either ($p_i=0$, $q_j=0$, $r_k\ne 0$) or ($p_i=0$, $q_j\ne 0$, $r_k=0$) or ($p_i\ne 0$, $q_j=0$, $r_k=0$). Without loss of generality we take ($p_i\ne 0$, $q_j=0$, $r_k=0$). Further, to achieve the maximal Mermin violation we must also have any one of the $p_i$'s be $1$ and $\rho^i_{BC}=|\mathcal{B}_{BC}\rangle\langle\mathcal{B}_{BC}|$, where $|\mathcal{B}_{BC}\rangle$ is a bipartite maximally entangled state that maximizes $\langle CHSH\rangle _{BC}$. To say more precisely, the maximal quantum value of CHSH is obtained by maximally entangled state of two qubits or a state equivalent to its local isometry \cite{McKague}.   

To see the requirement of mentioning the local dimension (of one of the subsystems), consider a tripartite probability distribution $\{p(abc|xyz)\}$ with Mermin value $2\sqrt{2}$ obtained from the quantum state $\rho_{ABC}=|\psi_A\rangle\langle\psi_A|\otimes |\mathcal{B}_{BC}\rangle\langle\mathcal{B}_{BC}|$ under some appropriate measurement choices (say) $\{x_0,x_1\}$ for Alice, $\{y_0,y_1\}$ for Bob, and $\{z_0,z_1\}$ for Charlie. Consider another distribution $\{p'(abc|xyz)\}$ with Mermin value $2\sqrt{2}$ and having a different quantum realization $\rho'_{ABC}=|\mathcal{B}_{AB}\rangle\langle\mathcal{B}_{AB}|\otimes|\psi_C\rangle\langle\psi_C|$ under some appropriate `primed' measurement choices. The probability distribution $\{k p(abc|xyz)+(1-k)p'(abc|xyz)~|~0<k<1\}$ allows a bi-separable quantum realization from the state $\sigma_{ABC}=k\rho_{ABC}\oplus(1-k)\rho'_{ABC}$ with a measurement setup $\{x_0\oplus x'_0,x_1\oplus x'_1\}$ for Alice and similarly for Bob and Charlie and the corresponding Mermin value is $2\sqrt{2}$. Clearly $\sigma_{ABC}$ is more general than the expression stated in the Theorem-\ref{theorem1}. However, if we restrict the local dimension of one of the subsystems then this generalized form of the state is no more allowed.
\end{proof}
According to the above theorem, for any quantum state (with one of the local subsystems dimension two), if it gives Mermin violation $2\sqrt{2}$, then the marginal statistics of two parties must be completely random. Therefore the local marginals can be used in witnessing the presence of genuine entanglement in semi-device independent manner. \\  
\emph{\bf Corollary}: Non-maximally mixed marginals for any of the two parties of a tripartite quantum correlation with Mermin  violation $\mathcal{M}=2\sqrt{2}$ certify genuine tripartite entanglement when at least one of the local Hilbert space dimension is two.

For explicit example, consider a noisy W-state, $\rho_v=v|W\rangle\langle W|+(1-v)\mathbb{1}/2\otimes\mathbb{1}/2$, where $|W\rangle=1/\sqrt{3}(|001\rangle+|010\rangle+|100\rangle)$. With certain noisy parameter values and suitably chosen measurements the tripartite correlations gives Mermin value $2\sqrt{2}$, with all the local expectation taking nonzero values (see Appendix.\ref{appendix1}). Similar example one can also contract from generalized-GHZ state, $|\psi_{GGHZ}\rangle=\cos\theta|000\rangle+\sin\theta|111\rangle$, with $0<\theta\le\pi/4$ (see Appendix.\ref{appendix2}). 

Consider now a quantity $Q$, made of from different Mermin polynomials, defined as follows, $Q:=\min \{Q_1,...,Q_9\}$, 
where $Q_1=|||\mathcal{M}_{000}-\mathcal{M}_{001}|-|\mathcal{M}_{010}-\mathcal{M}_{011}||-||\mathcal{M}_{100}-\mathcal{M}_{101}|
-|\mathcal{M}_{110}-\mathcal{M}_{111}|||$, and other $Q_i$'s are obtained by permutations. Here, $\mathcal{M}_{\alpha\beta\gamma}$ are different Mermin polynomials appearing in different Mermin inequalities, i.e., $\mathcal{M}_{\alpha\beta\gamma}:=
(\alpha\oplus\beta\oplus\gamma+1)\mathcal{M}_{\alpha\beta\gamma}^+
+(\alpha\oplus\beta\oplus\gamma)\mathcal{M}_{\alpha\beta\gamma}^-\le 2$, where $\mathcal{M}_{\alpha\beta\gamma}^+:=(-1)^{\gamma}\langle x_0y_0z_1\rangle+(-1)^{\beta}\langle x_0y_1z_0\rangle+(-1)^{\alpha}\langle x_1y_0z_0\rangle+(-1)^{\alpha\oplus\beta\oplus\gamma\oplus 1}\langle x_1y_1z_1\rangle$, and $\mathcal{M}_{\alpha\beta\gamma}^-:=
(-1)^{\alpha\oplus\beta\oplus1}\langle x_1y_1z_0\rangle+(-1)^{\alpha\oplus\gamma\oplus 1}\langle x_1y_0z_1\rangle+(-1)^{\beta\oplus\gamma\oplus 1}\langle x_0y_1z_1\rangle+\langle x_0y_0z_0\rangle$. Properties of this quantity and its use in studying \emph{non-classicality} of tripartite correlations have been discussed in \cite{Jebaratnam}. From the proof of Theorem-1 it becomes that the most general bi-separable states giving Mermin value $2\sqrt{2}$ and having at least two of the marginals non-maximally mixed, is of the form $\sigma^{bi.sep}_{ABC}=k_1|\psi_A\rangle\langle\psi_A|\otimes |\mathcal{B}_{BC}\rangle\langle\mathcal{B}_{BC}|+k_2|\psi_B\rangle\langle\psi_B|\otimes |\mathcal{B}_{AC}\rangle\langle\mathcal{B}_{AC}|+|k_3\psi_C\rangle\langle\psi_C|\otimes |\mathcal{B}_{AB}\rangle\langle\mathcal{B}_{AB}|$, and hence the corresponding tripartite probability distribution is of the form $P^{L_2}=k_1D_A\times PR^{\lambda}_{BC}+k_2D_B\times PR^{\lambda}_{AC}+k_2D_C\times PR^{\lambda}_{AB}$, with $\lambda=1/\sqrt{2}$, where $D$ denotes some deterministic box for single party and $PR^{\lambda}$ denotes noisy Popescu-Rohrlich (PR) correlation \cite{PR94}, i.e., $PR^{\lambda}:=\lambda PR+(1-\lambda)\mathcal{I}$, with $\mathcal{I}$ denoting the white noise \cite{self1}. Therefore, for any such correlation $P^{L_2}$, two of the Mermin polynomials $\mathcal{M}_{\alpha\beta\gamma}$ are equal while the rest are zero, which imply $Q$ to be zero. The implication of the above discussion can be expressed as the following theorem \cite{self2}.
\begin{theorem}\label{theorem2}
Non-maximally mixed marginals for any of the two parties of a tripartite quantum correlation with Mermin  violation $\mathbb{M}=2\sqrt{2}$ and $\mathcal{Q}>0$ certify genuine tripartite entanglement in a fully device-independent
way. 
\end{theorem} 

For the example of correlation given in Appendix.\ref{appendix1}  we find that $\mathcal{Q}>0$, which means that this correlation can not come from any bi-separable states even if there is no restriction on the dimensions of the local Hilbert spaces. 

\emph{Discussions}: In conclusion, we have shown how local marginals can help to detect tripartite entanglement in a semi-device independent and device independent manner. At this point a comparative discussion with a recent work \cite{Jebaratnam'2016} is relevant. In \cite{Jebaratnam'2016}, the author addressed the question of detecting genuine multipartite entanglement in steering scenario and introduced genuine steering inequalities.  Violation of any of these inequalities certifies genuine tripartite entanglement in semi-device-independent manner where the local Hilbert space dimension of two parties are assumed to be two. In comparison to this our semi-device independent scenario assumes dimension of only one of the subsystems. Another recent work \cite{Walter'13} by Walter et al. is also quite worthy to mention here. They have introduced a geometrical object, called \emph{entanglement polytope} which characterizes the eigenvalues of the single-particle states in any given class of entanglement. In case of pure, multi-particle quantum states, the authors have shown that the features of the global entanglement can be extracted from its local information alone. Unlike the present paper, the approach in \cite{Walter'13} is not device-independent. However, extending this idea in device-independent framework may be an interesting research direction.

\appendix

\section{Example- Noisy W state}\label{appendix1}
Here we take the example of a genuinely entangled state, which has optimal Mermin violation $2\sqrt{2}$ and demonstrate that one can witness the `genuineness' in entanglement by looking at the local statistics. 
Suppose, Alice, Bob and Charlie perform the projective qubit measurements $\{A_i\}_{i=0,1}$, $\{B_j\}_{j=0,1}$, and $\{C_k\}_{k=0,1}$. The noisy $W$ state with $v= 0.928585$ achieves the maximum Mermin value $2\sqrt{2}$ for the following measurement settings: 
\begin{align}
A_0 &=& -0.778908 \sigma_x - 0.219856 \sigma_y + 0.587337 \sigma_z,\nonumber\\ A_1 &=& 0.389816 \sigma_x + 0.11003 \sigma_y + 0.914296 \sigma_z,\nonumber\\ B_0 &=& -0.778908 \sigma_x - 0.219856 \sigma_y + 0.587337 \sigma_z,\nonumber\\ B_1 &=& 0.389816 \sigma_x + 0.11003 \sigma_y + 0.914296 \sigma_z,\nonumber\\ C_0 &=& -0.778908 \sigma_x - 0.219856 \sigma_y + 0.587337 \sigma_z,\nonumber\\ C_1 &=& 0.389816 \sigma_x + 0.11003 \sigma_y + 0.914296 \sigma_z\nonumber.
\end{align} 
But for the above measurement settings, all local expectation values are non-zero: $\langle A_0\rangle=\langle B_0\rangle=\langle C_0\rangle=0.18599$,
$\langle A_1\rangle=\langle B_1\rangle=\langle C_1\rangle=0.289527$. We also find that $\mathcal{Q}>0$ for this correlation which enables device independent certification of genuine entanglement.
\par 
Now one might wonder whether genuineness of entanglement can be certified in a device independent manner for the states with $v>0.928585$. Surely these states achieve Mermin violation greater than $2\sqrt{2}$ under suitable measurement settings and thus are detectable under the existing notion of DI GME using MI. But without optimizing the Mermin violation, if a state with $v>0.928585$ achieves Mermin violation equal to $2\sqrt{2}$ for some measurement settings, even then the non-vanishing local expectation values can certify genuine entanglement in a device independent way. 

\section{Example- GGHZ state}\label{appendix2}
Here we also demonstrate our claim with the example of GGHZ states.
The GGHZ state with $\theta = 0.4077$ gives Mermin violation $2\sqrt{2}$ for the following measurement settings:
\begin{eqnarray} 
A_0 &=& 0.0988727 \sigma_x + 0.980477 \sigma_y + 0.169967 \sigma_z,\nonumber\\ A_1 &=& 0.980477 \sigma_x  - 0.0988727 \sigma_y + 0.169967 \sigma_z,\nonumber\\ B_0 &=& 0.116457 \sigma_x + 0.978544 \sigma_y +  0.169967 \sigma_z,\nonumber\\ B_1 &=& 0.978544 \sigma_x - 0.116457 \sigma_y + 0.169967 \sigma_z,\nonumber\\ C_0 &=& -0.976125 \sigma_x -0.21721 \sigma_y,\nonumber\\ C_1 &=& -0.217125 \sigma_x + 0.975742 \sigma_y -0.0280037 \sigma_z.\nonumber
\end{eqnarray} 
Local expectation values of Alice and Bob are nonzero for the above measurement settings. 

{\bf Acknowledgment:} We would like to gratefully acknowledge fruitful discussions with Guruprasad Kar. MB thankfully acknowledges discussion with Sibasish Ghosh. AM acknowledges support from the CSIR project 09/093(0148)/2012-EMR-I. CJ acknowledges support through the Project SR/S2/LOP-08/2013 of the DST, Govt. of India.


\begin{thebibliography}{99}
\bibitem{self1} Note that for a given Mermin inequality violation $2\sqrt{2}$ achieved by the bi-separable state $\sigma^{bi.sep}_{ABC}$, only the noisy version of that particular PR-box $PR_{BC}$ (with suitable deterministic box $D_A$) which gives $4$ for the same Mermin inequality will appear in the decomposition and similar constraint holds for other two PR-boxes (i.e. $PR_{AB}$ and $PR_{AC}$) in the decomposition.
\bibitem{self2} One can also define a quantity, say $\tilde{Q}:=\min_{\alpha,\beta,\gamma,\alpha',\beta',\gamma'}|\mathcal{M}_{\alpha\beta\gamma}-\mathcal{M}_{\alpha'\beta'\gamma'}|$, with at least one pair of the primed and un-primed sub-indices are different, i.e., at least one of the relations $\alpha\ne\alpha'$, $\beta\ne\beta'$, or $\gamma\ne\gamma'$ must hold. Analogous to $Q$, this $\tilde{Q}$ also serves the purpose in Theorem-\ref{theorem2} similarly.
\end{thebibliography}
\end{document}